\documentclass[twocolumn,ruled,vlined,conference]{IEEEtran}
\usepackage[latin9]{inputenc}
\usepackage{array}
\usepackage{verbatim}
\usepackage{multirow}
\usepackage{algorithm2e}
\usepackage{amsmath}
\usepackage{amsthm}
\usepackage{graphicx}

\makeatletter

\providecommand{\tabularnewline}{\\}

\theoremstyle{plain}
\newtheorem{thm}{\protect\theoremname}
\theoremstyle{plain}
\newtheorem{prop}[thm]{\protect\propositionname}

\usepackage[final]{pdfpages}
\usepackage[belowskip=-4pt,aboveskip=0pt]{caption}
\usepackage{cite}
\usepackage{amsmath}
\usepackage{amsfonts}
\usepackage{amssymb}
\usepackage{amsthm}
\usepackage{multirow}
\usepackage{array}
\usepackage{slashbox}
\usepackage{makecell}
\usepackage{threeparttable}
\usepackage{algorithmic}
\usepackage{array}
\newcolumntype{I}{!{\vrule width 2pt}}
\usepackage{makecell}
\usepackage{threeparttable}

\usepackage{rotating}

\usepackage{pict2e}
\usepackage{cases}
\usepackage{enumitem}
\usepackage{footnote}
\makesavenoteenv{tabular}
\makesavenoteenv{table}

\usepackage{url}
\newcommand{\figref}{Figure }
\newcommand{\tabref}{Table }
\newcommand{\secref}{Section }

\newcommand{\propsname}[1]{Proposition #1}

\title{An Implementation of List Successive Cancellation Decoder with Large List Size for Polar Codes}

\author{ ChenYang Xia$^{\star}$, YouZhe Fan$^{\star}$, Ji Chen$^{\star}$, Chi-ying Tsui$^{\diamond}$, ChongYang Zeng$^{\dagger}$, Jie Jin$^{\dagger}$, and Bin Li$^{\dagger}$\\ $^{\star\diamond}$Department of Electronic and Computer Engineering, the HKUST, Hong Kong\\$^{\dagger}$Communications Technology Research Lab., Huawei Technologies, P. R. China\\$^{\star}$\{cxia, jasonfan, jchenbh\}@connect.ust.hk, $^{\diamond}$eestui@ust.hk \\$^{\dagger}$\{zengchongyang, steven.jinjie, binli.binli\}@huawei.com\\}

\@ifundefined{showcaptionsetup}{}{%
 \PassOptionsToPackage{caption=false}{subfig}}
\usepackage{subfig}
\makeatother

\providecommand{\propositionname}{Proposition}
\providecommand{\theoremname}{Theorem}

\begin{document}
\maketitle
\begin{abstract}
Polar codes are the first class of forward error correction (FEC)
codes with a provably capacity-achieving capability. Using list successive
cancellation decoding (LSCD) with a large list size, the error correction
performance of polar codes exceeds other well-known FEC codes. However,
the hardware complexity of LSCD rapidly increases with the list size,
which incurs high usage of the resources on the field programmable
gate array (FPGA) and significantly impedes the practical deployment
of polar codes. To alleviate the high complexity, in this paper, two
low-complexity decoding schemes and the corresponding architectures
for LSCD targeting FPGA implementation are proposed. The architecture
is implemented in an Altera Stratix V FPGA. Measurement results show
that, even with a list size of 32, the architecture is able to decode
a codeword of 4096-bit polar code within 150 $\mu$s, achieving a
throughput of 27Mbps.
\end{abstract}


\IEEEpeerreviewmaketitle

\section{Introduction\label{sec:introduction}}

As an emerging class of \emph{forward error correction} (FEC) codes
with a provably capacity-achieving capability, \emph{polar codes}
\cite{Arikan2009} attract a lot of research interests recently. To
decode the polar codes, \emph{list successive cancellation decoding}
(LSCD) \cite{Kai_List_2012,TIT_Tal_2015} was proposed, which outputs
$\mathcal{L}$ (called \emph{list size}) decoding paths by using
$\mathcal{L}$ parallel \emph{successive cancellation decodings}
(SCDs) \cite{TSP_Gross_2013,TSP_Fan_2014}. By concatenating the polar
codes with \emph{cyclic redundancy check} (CRC) codes \cite{TIT_Tal_2015,CL_Bin_2012}
and using the checksums to choose the most reliable path from the
list, LSCD with a large list size ($\mathcal{L}\geq16$) achieves
a similar or even better performance than other well-known FEC codes
\cite{TIT_Tal_2015}, such as low-density parity-check codes and turbo
codes. However, this comes at a high hardware cost as the complexity
scales with the list size $\mathcal{L}$. Thus, a low-complexity implementation
of the corresponding LSCD is very desirable.

The existing LSCD architectures \cite{Balatsoukas_TSP_2015,Bo_Yuan_TVLSI_2015},
which were designed for a small or medium list size ($\mathcal{L}\leq8$),
are not suitable for a large list size due to their high complexity
that is mainly due to two computational blocks. Firstly, several crossbars
are required for executing the \emph{list management} (LM) operation
\cite{Balatsoukas_TSP_2015} and they have complexity of $O(\mathcal{L}^{2})$.
Secondly, a sorter with $2\mathcal{L}$ inputs is needed to compare
and select the $\mathcal{L}$ best out of $2\mathcal{L}$ decoding
paths to keep the list size to $\mathcal{L}$ during the decoding
process. To reduce the logic delay, usually, a parallel sorter is
used \cite{Balatsoukas_TSP_2015}. However, this parallel sorter
has $O(\mathcal{L}^{2})$ comparators and hence dictates the clock
frequency and incurs high hardware complexity.

Recently, two \emph{field programmable gate array} (FPGA) implementations
of LSCD architectures were presented in \cite{crxiong_lehigh_sips_2016_fpgaemul,xliang_seu_globecomm_2016_ds},
which can be used as the emulation platforms for evaluating the performance
of polar codes. Due to the high complexity of LSCD, these platforms
cannot support an LSCD of $\mathcal{L}>4$. Moreover, to the best
of our knowledge, hardware implementation for LSCD with $\mathcal{L}=32$
has not been investigated in the literatures yet. In this work, we
first propose two low-complexity decoding schemes for the LSCD with
a large list size based on the analysis of the design constraints.
Then, an LSCD architecture using these schemes is developed and implemented
in an Altera FPGA. Measurement results show that our LSCD of $\mathcal{L}=32$
decodes a 4096-bit polar code within 150 $\mu$s to achieve a 27Mbps
throughput.

\section{Preliminaries\label{sec:review} }

\subsection{Code Construction\label{subsec:code-constuction}}

Considering a polar code with length $N=2^{n}$. Its generator matrix,
$\textbf{F}^{\otimes n}$, is the $n^{th}$ Kronecker power of $\textbf{F}=\begin{bmatrix}1 & 0\\
1 & 1
\end{bmatrix}$. Source word $\textbf{u}$ and code word $\textbf{x}$ are two $N$-bit
binary vectors related by $\textbf{x}=\textbf{u}\cdot\textbf{F}^{\otimes n}$.
The bits in $\textbf{u}$ have different reliabilities. The indices
of the $K$ most reliable bits compose the information set $\mathcal{A}$
while its complement $\mathcal{A}^{c}$ is called the frozen set.
Accordingly, $u_{i}\text{s}\,(i\in\mathcal{A})$ are called information
bits and are used to deliver message; while the rest are called frozen
bits and fixed to 0. The code rate is thus defined as $R=K/N$. When
an $r$-bit CRC code is concatenated, the last $r$ information bits
are used to deliver the CRC checksums of the other $K-r$ information
bits.

\subsection{List Successive Cancellation Decoding\label{subsec:review_lscd}}

\begin{figure}
\centering\subfloat[]{\includegraphics{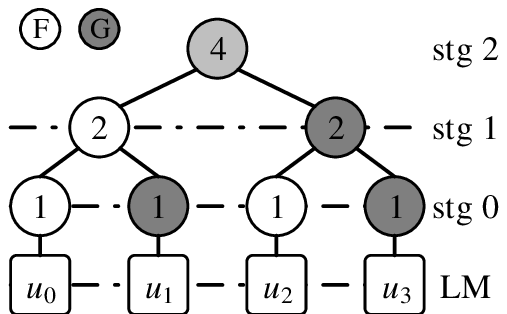}

}\subfloat[]{\includegraphics{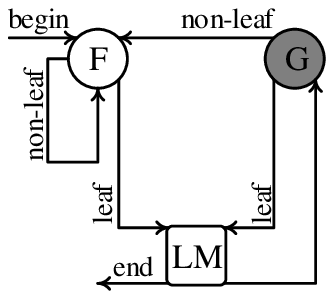}

}\caption{(a) scheduling tree of polar codes of $N=4$ and (b) the corresponding
state transfer diagram.}
\label{fig:codec}
\end{figure}

As shown in \figref \ref{fig:codec}(a), the LSCD is made up of $\mathcal{L}$
copies of SCD operations (each described by the full binary tree)
and the LM operations (represented by the squares).

The SCD operation is a depth-first traversal of the full binary tree
with $n+1$ stages which is also called a scheduling tree. The channel
\emph{log-likelihood ratios} (LLRs), $L_{i}=\textup{log}(\textup{Pr}(\textbf{y}|0))-\textup{log}(\textup{Pr}(\textbf{y}|1)),i\in[0,N-1]$,
are the inputs at the root node of the scheduling tree, where \textbf{$\textbf{y}$}
is the channel output on \textbf{$\textbf{x}$}. The left and right
children of a node are called F- and G-node whose functions are\footnote{\label{note1}The exact forms of these functions are non-linear. To
have an efficient hardware implementation, approximate forms, \eqref{eqn:f_func}
and \eqref{eqn:pmu_update}, are used\cite{TSP_Gross_2013,Balatsoukas_TSP_2015}.}
\begin{align}
L_{F}(L_{a},L_{b}) & =(\textup{sgn}(L_{a})\oplus\textup{sgn}(L_{b}))\cdot\textup{min}(|L_{a}|,|L_{b}|),\label{eqn:f_func}\\
L_{G}(\hat{s},L_{a},L_{b}) & =(-1)^{\hat{s}}L_{a}+L_{b},\label{eqn:g_func}
\end{align}
respectively, where $L_{a}$ and $L_{b}$ are the inputs of the both
functions from the previous stage. $\hat{s}$ in \eqref{eqn:g_func}
is a binary input called partial-sum, which is calculated from the
bits already decoded up to the corresponding G-node. In the LSCD,
all the $\mathcal{L}$ copies of SCDs are executed in parallel.

An LM is executed after a leaf node is reached by the SCDs. Assuming
that after $\hat{u}_{i-1}$ is decoded, the list is full of $\mathcal{L}$
paths and each path has a different decoded sub-vector $[\hat{u}_{0},...,\hat{u}_{i-1}]\in\{0,1\}^{i}$.
A \emph{path metric} (PM), $\gamma_{i-1}^{l}$, is associated with
each path to represent its reliability. When $\hat{u}_{i}$ is decoded,
the LM of $\hat{u}_{i}$ is executed in two steps. First, each path
is expanded into two with $\hat{u}_{i}$ instantiated to 0 and 1,
respectively. For a path $l$, its \emph{path metric update} (PMU)
is\footnotemark[1]
\begin{equation}
\begin{cases}
\begin{array}{l}
\gamma_{i}^{2l}=\gamma_{i-1}^{l},\\
\gamma_{i}^{2l+1}=\gamma_{i-1}^{l}+\left|\Lambda_{i}^{l}\right|,
\end{array} & \begin{array}{l}
\textrm{if}~\hat{u}_{i}=\Theta\left(\Lambda_{i}^{l}\right),\\
\textrm{if}~\hat{u}_{i}=1-\Theta\left(\Lambda_{i}^{l}\right),
\end{array}\end{cases}\label{eqn:pmu_update}
\end{equation}
where $\gamma_{i}^{2l}$ and $\gamma_{i}^{2l+1}$ are the PMs of the
two expanded paths and $\Lambda_{i}^{l}$ is the output LLR of the
$i^{th}$ leaf node. The hard decision is made by $\Theta\left(x\right)$=$(x<0)$.
If the number of paths exceeds $\mathcal{L}$ after the path expansion,
the \emph{list pruning operation} (LPO) is executed to find the $\mathcal{L}$
smallest PMs and keep them as the survival paths. Note that if $i\in\mathcal{A}^{c}$,
only one of the equations is executed and the LPO is not needed.

\subsection{Problems in the Existing LSCD Architectures\label{subsec:soa_lscd} }

\begin{table}
\centering

\begin{threeparttable}

\caption{Crossbar complexity for 4096-bit polar codes on Altera 5SGXEA7N2F45C2
(available ALMs: 234,720) }

\begin{tabular}{c||c|c|c|c|c}
\hline 
List size & 2 & 4 & 8 & 16 & 32\tabularnewline
\hline 
Req. ALMs & 10,240 & 15,360 & 87,040 & 414,720 & 1,479,680\tabularnewline
\hline 
\end{tabular}

\label{tab:complexity_crb}

\end{threeparttable}
\end{table}

Based on the algorithms presented above, several LSCD architectures
were proposed \cite{Balatsoukas_TSP_2015,Bo_Yuan_TVLSI_2015,crxiong_lehigh_sips_2016_fpgaemul,xliang_seu_globecomm_2016_ds,JSAC_FAN_2015}.
One common feature of them is that some $\mathcal{L}\times\mathcal{L}$
crossbars are needed to align the data in the $\mathcal{L}$ blocks
of SCD hardware according to the LM results. \tabref\ref{tab:complexity_crb}
shows the synthesis results of the crossbars used in the architecture
of \cite{JSAC_FAN_2015}. Here, an 8-bit quantization is used for
the LLRs. It can be seen the complexity scales far beyond $O(\mathcal{L})$
for a given polar code and the required resources far exceed the logic
resources, i.e. \emph{adaptive logic modules} (ALMs), available in
the FPGA for the LSCD with large list sizes. 

Another issue is the implementation of the sorter which is required
for finding the smallest $\mathcal{L}$ PMs after each path expansion.
According to \cite{Balatsoukas_TSP_2015}, the delay and complexity
of radix-2$\mathcal{L}$ sorter are significantly increased with $\mathcal{L}$.
The complexity of this sorter is further increased if a low-latency
LM scheme, such as \emph{multi-bit decoding} (MBD) \cite{Bo_Yuan_TVLSI_2015},
is used.

From the above discussion, implementing the architecture of LSCD with
a large list size on hardware, especially in a resource-limited device
such as an FPGA, is a very challenging task. In the following sections,
we will present some schemes to reduce the complexity of LSCD. 

\section{Low-complexity decoding schemes for LSCD}

\subsection{Parallel-F Serial-G Computation\label{subsec:pfsg} }

\begin{figure}[t]
\centering 

\includegraphics[width=88mm]{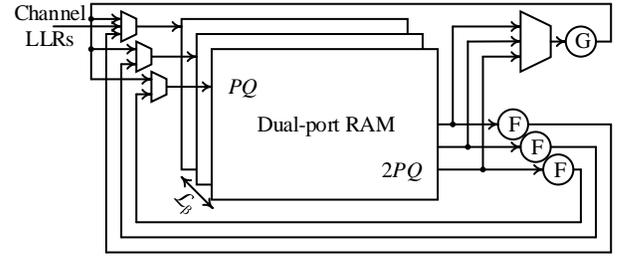} \caption{The structure of parallel-F serial-G computation.}
\label{fig:arc_pfsg} 
\end{figure}

From \secref \ref{subsec:soa_lscd}, it is beneficial to avoid using
crossbars in the architecture of LSCD with a large list size. A straightforward
method is to integrate $\mathcal{L}$ blocks of LLR memories into
a single memory and evaluating the SCD functions of each path serially.
By doing so, the required operands are obtained by accessing the memory
in the right locations. However, since the $\mathcal{L}$ SCDs are
executed serially, the decoding latency is $\mathcal{L}$ times that
of the traditional SCD. To reduce this latency, the following proposition
related to the LSCD is used.
\begin{prop}
\label{prop:crossbar} When the F-nodes are computed, the memories
and PE arrays are one to one corresponding and the crossbars do not
need to permute any data; only when the G-nodes are visited, crossbars
need to permute the data from the memories. 
\end{prop}
\begin{proof} This can be easily proved from the state transfer diagram,
as shown in \figref \ref{fig:codec}(b), which shows the execution
order of the F-nodes, G-nodes and LMs.\end{proof}

Based on \propsname{\ref{prop:crossbar}}, a \emph{parallel-F serial-G}
(PFSG) computation scheme is proposed. All the F-functions are calculated
in parallel for all the paths as the crossbar is not needed in this
situation and a direct connection between the corresponding memory
and PE array can already support the calculations. In contrast, the
G-functions of each path are serially evaluated to avoid using crossbars.
As the latency for evaluating these two kinds of functions are the
same in the SCD operation, the latency of LSCD using PFSG computation
is $\frac{\mathcal{L}+1}{2}$ times that of the traditional SCD, which
is reduced by almost one half comparing with that of straightforward
mapping for large $\mathcal{L}$.

The corresponding PFSG structure is shown in \figref \ref{fig:arc_pfsg}.
Each block of RAM is implemented with a dual-port RAM with a $2PQ$-bit
read port and a $PQ$-bit write port, where $Q$ is the number of
quantization bits for the LLRs. $\mathcal{L}+1$ groups of $P$ processing
elements are used in this structure. One group is for the G-nodes,
whose inputs are selected by an $\mathcal{L}$-to-1 multiplexer. The
others are for the F-nodes, which can calculate the F-functions for
$\mathcal{L}$ paths simultaneously.

It is noted when $\mathcal{L}$ is large, the utilization of RAMs
is temporarily low as only the data from one of the $\mathcal{L}$
blocks of RAMs are valid in each cycle. So, in the real implementation,
the number of blocks of RAMs can be reduced from $\mathcal{L}$ to
$\mathcal{L}_{\beta}$, which is a power of 2, and each block of RAM
stores the LLRs of $\mathcal{L}/\mathcal{L}_{\beta}$ paths. By choosing
a proper $\mathcal{L}_{\beta}$, the balance between the complexity
and the latency can be achieved.

\subsection{Low-Complexity List Management\label{subsec:lblm}}

In this section, a simplified LM operation of LSCD is proposed to
reduce the computational complexity. To avoid the long latency brought
by the G-nodes in the PFSG computation, the proposed method is based
on the MBD. Specifically, the MBD simultaneously decodes all the $M$
bits of a sub-tree rooted at stage $m$, where $M=2^{m}$. Let $\gamma_{\text{in}}$
be the PM of one survival path and $\gamma_{\text{out}}^{\text{MBD}}$
be the PM of one of its expanded paths, then the PMU of MBD is
\begin{equation}
\gamma_{\text{out}}^{\text{MBD}}=\gamma_{\text{in}}+{\displaystyle \sum}_{i=0}^{M-1}(v_{i}\oplus\Theta(L_{i}))\cdot|L_{i}|,\label{eqn:mbd}
\end{equation}
where $[L_{0},...,L_{M-1}]$ are the output LLRs at the root node
of the sub-tree and $[v_{0},...,v_{M-1}]=[\hat{u}_{0},...,\hat{u}_{M-1}]\cdot F^{\otimes m}$.
There are at most $2^{M}$ combinations of $v_{i}$s and hence at
most $2^{M}$ paths are expanded from each survival path, which incurs
a high complexity to the LPO even when $M$ is small. 

To reduce the complexity, we combine one of our previously proposed
algorithms, \emph{selective expansion} (SE) \cite{JSAC_FAN_2015},
with the MBD. The SE efficiently reduces the number of the expanded
paths by partitioning the information set $\mathcal{A}$ into an unreliable
set $\mathcal{A}_{u}$ and a reliable set $\mathcal{A}_{r}$ based
on the reliability of each information bit. The path expansions corresponding
to the bits belonging to $\mathcal{A}_{r}$ do not need to be executed.
We call the combined method \emph{low-complexity list management}
(LCLM). Supposing there are $M_{u}$ unreliable bits and $M_{r}$
reliable bits in a $M$-bit sub-tree. For a given set of values of
the unreliable bits, the PMU of one of the expanded paths is calculated
as
\begin{equation}
\gamma_{\text{out}}^{\text{LCLM}}=\text{min}_{u_{j}\in\{0,1\},\,j\in\mathcal{A}_{r}}(\gamma_{\text{out}}^{\text{MBD}}),\label{eqn:lclm_pmu}
\end{equation}
where $\gamma_{\text{out}}^{\text{MBD}}$s are obtained from \eqref{eqn:mbd}.
The minimum in \eqref{eqn:lclm_pmu} is selected over the $2^{M_{r}}$
$\gamma_{\text{out}}^{\text{MBD}}$s. To expand each survival path,
\eqref{eqn:lclm_pmu} needs to be calculated $2^{M_{u}}$ times as
$2^{M_{u}}$ paths will be generated from the path expansion. Finally,
LPO is used to select the $\mathcal{L}$ best paths from the $2^{M_{u}}\cdot\mathcal{L}$
expanded paths.

The LCLM expands fewer paths and hence achieves a lower complexity
than the MBD. Also, \propsname{\ref{props:lblm}} guarantees the
decoding performance of LCLM is not worse than that of SE. 
\begin{prop}
\label{props:lblm} For a given $\gamma_{\text{in}}$ and $u_{i}\text{s}\,(i\in\mathcal{A}_{u})$
in an $M$-bit tree, the updated PMs of LCLM and SE satisfy $\gamma_{\text{out}}^{\text{LCLM}}\leq\gamma_{\text{out}}^{\text{SE}}$. 
\end{prop}
\begin{proof} For the given $u_{i}\text{s}\,(i\in\mathcal{A}_{u})$,
the corresponding $\gamma_{\text{out}}^{\text{SE}}$ equals to one
of the $\gamma_{\text{out}}^{\text{MBD}}$s calculated by \eqref{eqn:mbd}.
So \eqref{eqn:lclm_pmu} ensures the validity of \propsname{\ref{props:lblm}}.
\end{proof}

An LSCD tries to find the best $\mathcal{L}$ paths with the locally
smallest PMs. \propsname{\ref{props:lblm}} ensures that the paths
generated by the LCLM is not worse than those by the SE. Hence, the
error performance of LCLM is at least as good as that of SE.

\section{Implementation Results\label{sec:impl} }

\subsection{The Implementation of the Proposed LSCD Architecture\label{sec:arc}}

\begin{figure}[t]
\centering \includegraphics[width=88mm]{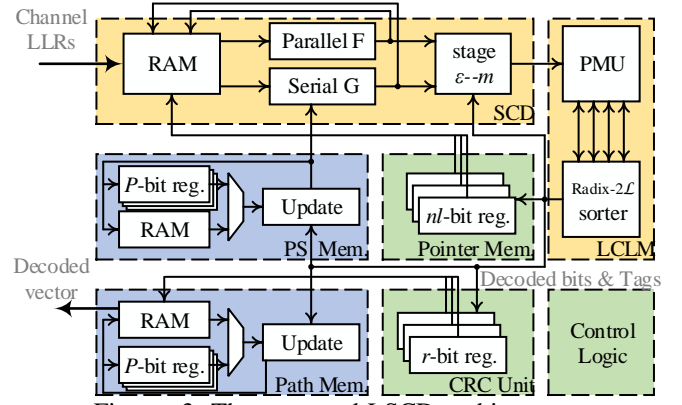} \caption{The proposed LSCD architecture.}
\label{fig:arc_proposed} 
\end{figure}

The implementation of the proposed LSCD architecture is shown in \figref
\ref{fig:arc_proposed}, which mainly includes seven blocks.

The SCD module is used to compute the F- and G-nodes to obtain the
LLR outputs of the stages higher than stage $m-1$. We further divide
these stages into high stages (higher than a pre-determined stage
$\epsilon$) and low stages (the rest). The high stages are calculated
with the PFSG structure. The low stages are calculated in a parallel
fashion as the PFSG brings a large latency overhead for these stages.
Specifically, one memory is used to store the LLRs of all the paths
and only one SCD hardware for the low stages is connected with it.
Such structure is duplicated $\mathcal{L}$ times and the computations
of all the paths can be executed simultaneously without a crossbar.
The LCLM module receives the LLR outputs at stage $m$ from the SCD
module. Here, a radix-$2\mathcal{L}$ parallel sorter is used. If
$M_{u}>1$ in a sub-tree, the $2\mathcal{L}$-to-$\mathcal{L}$ sorting
is executed multiple times in serial to find the best $\mathcal{L}$
paths. The LCLM greatly reduces the number of expanded paths, so the
latency for sorting is moderate. The outputs of the LCLM module include,
for each path, $M$ decoded bits and a tag, indicating which survival
path the expanded path is extended from.

The partial-sum memory and the path memory are used to store and update
the partial-sums and the decoded vectors of the $\mathcal{L}$ paths,
respectively. These memories are only activated when a G-node is calculated.
Therefore the crossbars originally required in these two blocks in
the existing architectures are not needed as the PFSG computation
is used. A two-staged memory structure similar to the folded partial-sum
network in \cite{TSP_Fan_2014} is used. The other parts, including
the pointer memory, the CRC unit and the control logic, are similar
to their counterparts in the existing architectures \cite{Balatsoukas_TSP_2015,JSAC_FAN_2015}.

\subsection{Implementation and Measurement Results in the FPGA}

\begin{table}
\centering

\begin{threeparttable}

\caption{The LSCD parameters for FPGA implementation.}
{\footnotesize{}\label{tab:par_code-1}}\setlength{\tabcolsep}{8pt}%
\begin{tabular}{c|c|>{\centering}p{0.6cm}|c|c|c|>{\centering}p{0.6cm}}
\hline 
$N$  & $K$  & $r^{\text{a}}$ & \multicolumn{3}{c|}{CRC generator polynomial} & $\mathcal{L}$\tabularnewline
\hline 
4096  & 2048  & 24  & \multicolumn{3}{c|}{0x864cfb } & 32\tabularnewline
\hline 
\hline 
$\mathcal{L}_{\beta}$ & $P$ & $Q$ & $Q_{\text{PM}}$ & $\eta$@SNR=2dB\tnote{b} & $m$ & $\epsilon$\tabularnewline
\hline 
4 & 128 & 8 & 9 & 0.3 & 2 & 3\tabularnewline
\hline 
\end{tabular} 

{\footnotesize{} \begin{tablenotes} \item [a] The effective code
rate is $R=\frac{K-r}{N}=0.494$. \end{tablenotes}}{\footnotesize \par}

{\footnotesize{} \begin{tablenotes} \item [b] Following the method
and notation of \cite{JSAC_FAN_2015}, $\eta$ determines $\mathcal{A}_{u}$
for SE. \end{tablenotes}}{\footnotesize \par}

\label{tab:par_code}

\end{threeparttable}
\end{table}
\begin{table}[t]
\centering{\footnotesize{} \caption{Hardware usage of the LSCD architecture in FPGA.}
 }%
\begin{tabular}{c|c|c|c}
\hline 
\multirow{1}{*}{} & \multirow{1}{*}{ALMs } & \multirow{1}{*}{Registers } & RAM blocks\tabularnewline
\hline 
LSCD usage  & 67,211  & 31,247  & 1,122\tabularnewline
\hline 
FPGA capacity  & 234,720  & 939,000  & 2,560\tabularnewline
\hline 
Utilization  & 28.63\%  & 3.33\%  & 43.82\% \tabularnewline
\hline 
\end{tabular}{\footnotesize{}\label{tab:hw_usage}}
\end{table}
\begin{table}[t]
\centering

\begin{threeparttable}{\footnotesize{}\caption{Comparison of the implementation results of several FPGA-based LSCD
architectures.}
\setlength{\tabcolsep}{10pt}}%
\begin{tabular}{c|>{\centering}p{1.6cm}|>{\centering}p{1.5cm}|>{\centering}p{1.5cm}}
\hline 
 & Proposed & \cite{crxiong_lehigh_sips_2016_fpgaemul} & \cite{xliang_seu_globecomm_2016_ds}\tabularnewline
\hline 
\multirow{2}{*}{FPGA Device\tnote{c}} & Altera & Xilinx & Altera\tabularnewline
 & Stratix V & Kintex 7 & Stratix V\tabularnewline
\hline 
(N,$\mathcal{L}$) & (4096,32) & (1024,4) & (1024,4)\tabularnewline
\hline 
ALMs(A)/LUTs(X)\tnote{d} & 67,211 & 142,961 & 101,160\tabularnewline
\hline 
Registers & 31,247 & 19,795 & 13,544\tabularnewline
\hline 
RAM (Mbits) & 22.440 & 4.404 & 0\tabularnewline
\hline 
Clock rates (MHz) & 107 & 42.66 & N/A\tabularnewline
\hline 
Throughput (Mbps) & 27.35 & 115 & N/A\tabularnewline
\hline 
\end{tabular}{\footnotesize \par}

{\footnotesize{} \begin{tablenotes} \item [c] All the FPGAs are manufactured
on 28nm process technology. \end{tablenotes}}{\footnotesize \par}

{\footnotesize{} \begin{tablenotes} \item [d] An ALM on Altera FPGA
can be used as a 6-input LUT. \end{tablenotes}}{\footnotesize \par}

{\footnotesize{}\label{tab:timing_cmp}}{\footnotesize \par}

\end{threeparttable}
\end{table}
\begin{figure}[t]
\vspace{-10pt}\includegraphics[width=8.8cm]{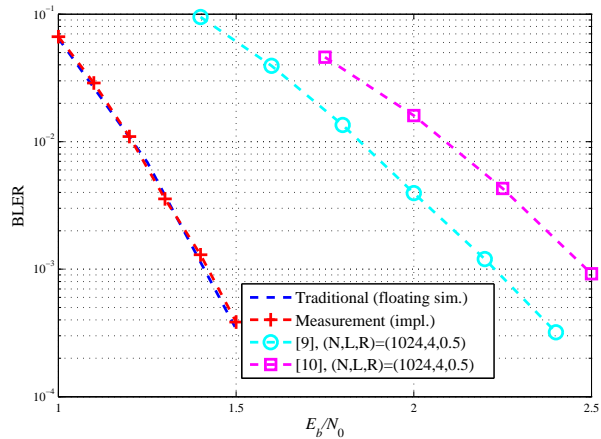}
\caption{The BLERs of different LSCDs.}
\label{fig:hardware_sim} 
\end{figure}

To demonstrate the performance of the FPGA implementation of the above
LSCD architecture, we implement it in an Altera Stratix V 5SGXEA7N2F45C2
FPGA. \tabref \ref{tab:par_code} summarizes the parameters of the
target polar codes and the implemented decoder. The LSCD architecture
is mapped on the FPGA with a clock frequency of 107MHz. The decoding
latency of the LSCD is 16019 cycles for one codeword, translating
into 149.71 $\mu$s under the target clock frequency. The hardware
usage of our LSCD under the specified constraint is shown in \tabref
III. Among all the resources, the RAM blocks (each with 20 kbits)
have the highest usage, 22.44 Mbits, which is much higher than the
theoretical value of about 2 Mbits. This is because the port width
of one RAM block is limited. To guarantee the calculation parallelism,
relatively wide port widths are used and multiple RAM blocks are then
needed, leading to the high usage of RAM blocks.

\tabref IV compares our LSCD with other FPGA-based LSCD architectures
in the literatures \cite{crxiong_lehigh_sips_2016_fpgaemul,xliang_seu_globecomm_2016_ds}.
Our architecture can support a longer code length and a much larger
list size with even lower utilization of logic resources. The memory
resources used per path are less than those of \cite{crxiong_lehigh_sips_2016_fpgaemul}.
Though the memory usage of the architecture in \cite{xliang_seu_globecomm_2016_ds}
is lower, without any reported timing results, it is not easy to determine
which architecture makes a better tradeoff between the complexity
and the latency. The comparison results indicate the proposed low-complexity
schemes are very efficient. At the same time, though the latency of
our LSCD is supposed to scale linearly with the list size, the throughput
degradation is less than linear. Also, for the other two architectures,
it is not feasible to use them to implement LSCD with a large list
size in an FPGA.

Finally, the measured \emph{block error rate} (BLER) of the implemented
LSCD is shown in \figref \ref{fig:hardware_sim}. For this measurement,
an encoder and an additive white Gaussian noise channel are also implemented
on-chip. As reference, the simulated BLER of the traditional LSCD
with floating-point is also shown. It can be seen that our LSCD functions
well and the performance degradation is less than 0.05dB at a BLER
of $10^{-3}$, which is the target BLER of a typical cellular communication
system. Also, comparing with the testing results presented in \cite{crxiong_lehigh_sips_2016_fpgaemul},
a performance gain of about 0.8dB is achieved at this BLER. 

\section{Conclusion\label{sec:conclusion} }

In this work, two low-complexity decoding schemes, namely PFSG and
LCLM schemes, are proposed for the LSCD with a large list size, and
the corresponding architecture for FPGA is developed and implemented.
The measurement results show that the proposed LSCD ($\mathcal{L}$=32)
has low hardware usage with negligible error performance degradation. 


\end{document}